\documentclass[a4paper]{article}

\usepackage{fullpage, url,amsmath,amsfonts,amssymb,tedmath,mathtools}
\newcommand{\ord}{$\rm{OrderMod}$}

\date{}
\title{Maximum distance separable codes to order} 
\author{Ted Hurley\footnote{National University of Ireland
Galway. Ted.Hurley@NuiGalway.ie} , Donny Hurley\footnote{Institute of
Technology, Sligo. hurley.donny@itsligo.ie },   Barry
Hurley\footnote{barryj\_2000@yahoo.co.uk}} 
\begin{document}
\maketitle
\begin{abstract}

 
Maximum distance separable (MDS) are
 constructed to required specifications. The codes are explicitly given  over
 finite fields with  efficient encoding and decoding
 algorithms. Series of such codes over
 finite fields with
 ratio of distance to length approaching $(1-R)$ for given $R, \, 0 < R < 1$
 are derived. For given rate $R=\frac{r}{n}$,  with $p$ not dividing 
 $n$,  series of codes over finite fields of characteristic $p$
 are constructed     
 such that the ratio of the distance to the
 length approaches $(1-R)$. For a given field $GF(q)$ MDS codes of the
 form $(q-1,r)$ are constructed for any $r$. 
The codes are encompassing, easy to
 construct with efficient encoding 
 and decoding algorithms of complexity  $\max\{O(n\log n), t^2\}$, where
 $t$ is the error-correcting capability of the code.      
 
\end{abstract}
\section{Introduction}

  Coding theory is at the heart of modern day communications. Maximum distance separable, MDS, codes
 are at the heart of coding theory. Data needs
 to be transmitted {\em safely } and sometimes securely. Best rate and error-correcting capabilities are the aim, and MDS codes can meet the requirements;  they  correct the maximum number of errors for given length and dimension. 



General methods  for constructing  MDS codes over finite fields are  given in 
 Section \ref{construct} following \cite{hurley,hurley1,hurleycomp}. 
 The codes  are explicitly constructed over finite fields with  efficient encoding and decoding
 algorithms of complexity $\max\{O(n\log n), O(t^2)\}$,
 where $t$ is the error-correcting capability. 
These are exploited. For given $\{n,r\}$  
  MDS $(n,r)$ codes are constructed over finite fields with  characteristics
 not dividing  $n$, section \ref{general2}. For given rate and given
error-correcting capability series of MDS codes to these specifications are
constructed over finite fields, section \ref{required}. 
For given
rate $R$, $0<R<1$,  series of MDS codes are constructed over finite fields in
which the ratio of the distance by the length approaches $(1-R)$, section \ref{infinite}. 



For a given finite field $GF(q)$,  MDS $(q-1,r)$ codes of different types are
 constructed over $GF(q)$  for any given $r, \, 1\leq r \leq (q-1)$, section \ref{length}. The
 codes are explicit with efficient encoding and decoding algorithms as
 noted. In addition for each $n/(q-1)$, MDS codes of length $n$ and
 dimension  $r$ are constructed over $GF(q)$ for any given $r, \, 1 \leq r \leq n$. In particular for $p$ a prime,  MDS $(p-1, r)$ codes are constructed in $GF(p) = \Z_p$ in which case the arithmetic is modular arithmetic which works smoothly and very  efficiently.

For given $R = \frac{r}{n}$, $0<R<1$,  with  $p\not\vert\,n$,  
 series of codes over finite fields  of characteristic $p$  are
constructed in which the ratio of the distance to the length approaches
$(1-R)$, section \ref{charp}. Note $0<R<1$ if and only if $0<(1-R)<1$.     In particular such series are constructed in fields of characteristic $2$ for cases where  the denominator  $n$ of the given rate  is odd. 

Series of MDS codes over prime fields $GF(p) = \Z_p$ are constructed such that the ratio of the distance to the length approaches $(1-R)$ for given $R, 0<R<1$; in these cases the arithmetic is modular
arithmetic which is extremely  efficient and easy to implement, section \ref{primes}. 

Samples are given in the different sections and an example is given on
the workings of the decoding algorithms in section \ref{err}. The
explicit examples given  need to be of reasonably small  size for
display here but in general there is no restriction on the length or dimension in practice. 

Explicit efficient encoding and decoding algorithms of complexity $\max\{O(n\log n), O(t^2)\}$  exist for the codes and this  is explained in section \ref{complexity}.   

The codes  are encompassing  and excel known used and practical codes.
See for example section \ref{prime} for the following:  
MDS codes of the form $(255,r)$ for
any $r, 1\leq r \leq 255$ are constructed over $GF(2^8)$.  They are constructed explicitly and have efficient encoding and decoding algorithms which reduce to  finding a solution of a Hankel $ t\ti (t+1)$ system, where $t$ is the error-correcting capability, and matrix multiplications by a Fourier matrix.  These can be
compared to the Reed-Solomon codes over $GF(2^8)$. The method extends easily to the formation of MDS codes of the form $(511,r)$ for any $r, 1
\leq r\leq 511$ over $GF(2^9)$, and then further to MDS codes $(2^k-1,r)$
over $GF(2^k)$.  
 Codes over prime fields are particularly nice and as an example
$(256,r)$ codes are constructed over $GF(257)=\Z_{257}$. The arithmetic
is modular arithmetic over $\Z_{257}$; these perform better than the 
$(255,r)$ RS codes over $GF(2^8)$. These can also easily  be extended for larger primes as for example $(10008,r)$ MDS codes over $GF(10009)$. 

In general: For  any prime $p$, $(p-1,r)$ codes
over $GF(p)=\Z_{p}$ are constructed for any $r, 1\leq r \leq (p-1)$; for any $k$, $(2^k-1, r)$ codes are
constructed over $GF(2^k)$ and any $r, 1\leq r \leq (2^k-1)$.  As
already noted the constructed codes have  (very) efficient encoding and decoding algorithms.

The encoding and decoding methods involve multiplications by a Fourier matrix and finding a solution to a Hankel $t \ti (t+1)$ system, where $t$ is the error-correcting capability of the MDS code.

Background on coding theory and field theory may be
found  in \cite{blahut}, \cite{mac}  or \cite{mceliece}.
An $(n,r)$ linear code is a linear code of length $n$ and dimension $r$;
the {\em rate} of the code is $\frac{r}{n}$. An $(n,r,d)$ linear code is a code of length
$n$, dimension $r$ and (minimum) distance $d$. The code is an MDS code
provided $d=(n-r+1)$, which is the maximum distance an $(n,r)$ code can
attain. The error-capability of $(n,r,d)$ is $t=\floor{\frac{d-1}{2}}$ which is the
maximum number of errors the code can correct successfully. 
The finite field of order $q$ is denoted by $GF(q)$ and of necessity $q$ is a power of a prime.   

The
 codes are generated by the unit-derived method -- see
 \cite{hur1,hur2,hur0}  -- by choosing  rows in sequence of
 Fourier/Vandermonde  matrices over finite fields  following the methods
 developed in \cite{hurley,hurley1}. They are easy to implement, explicit and with efficient encoding and decoding algorithms of complexity $\max\{(O\log n), O(t^2)$ where $t$ is the error-correcting capability. 

 \subsubsection{Particular {\em types} of MDS codes}\label{types1}
 Different {\em types} of MDS codes, such as Quantum or Linearly
 complementary dual (LCD) codes,   can be constructed based on general
 schemes; see section \ref{further} for references on these
 developments. This section also notes a reference to using these
 types of error-correcting codes in solving underdetermined
 systems of equations for {\em compressed sensing} applications. 
 
 
\section{Constructions}\label{construct} 

\subsection{Background material}\label{background} 

In \cite{hur1,hur0} systems of {\em unit-derived codes} are developed;
 a suitable version  in book chapter form is available at \cite{hur2}. 
 In summary the unit-derived codes  are obtained  as follows. Let $UV= I_n$ in a ring.
  Let $G$ be the $r\ti n$ matrix generated by choosing any $r$ rows of $U$
 and let $H\T$ be the $n \ti (n-r) $ matrix obtained from $V$ by
 eliminating the corresponding columns of $V$. Then $G$ generates an
 $(n,r)$ code and $H$ is the the check matrix of the code. The system
 can be considered in format as  $GH\T = 0_{r\ti (n-r)}$. 

 When the first rows are chosen as generator matrix,  the process may be presented as follows.
Let $UV = I_n$ with $U= \begin{pmatrix} A \\ B \end{pmatrix}, V= (C,D)$
where $A$ is an $r\ti n$ matrix, $B$ is an $(n-r) \ti n$ matrix, $C$ is
an $n\ti r$ matrix and $D$ is an $n\ti (n-r)$ matrix. Then 
$UV=I$ gives $\begin{pmatrix} A \\ B \end{pmatrix} (C,D) = \begin{pmatrix} I_r & 0 \\ 0 & I_{n-r}\end{pmatrix}$ and so in
particular this gives $AD = 0_{r\ti (n-r)}$. The matrices have full
rank. Thus with $A$ as the generating matrix of an $(n,r)$ code it is
seen that $D\T$ is the check matrix of the code. 

By explicit row selection,  the process is as follows. 
 Denote  the rows of $U$ in order by $\langle e_0,e_1, \ldots, e_{n-1}\rangle$ and the columns of $V$ in order by $\langle f_0, f_1, \ldots, f_{n-1}\rangle$.
Then $\begin{pmatrix}e_0 \\ e_1 \\ \vdots \\ e_{n-1}\end{pmatrix}(f_0,f_1, \ldots, f_{n-1}) = I_n$.
From this it is seen  that $e_if_i= 1, e_if_j = 0, i\neq j$.

Thus if $G=\begin{pmatrix}e_{i_1} \\ e_{i_2} \\ \vdots \\
	   e_{i_r}\end{pmatrix}$ (for distinct  $e_{i_k}$)
and $H\T=(f_{j_1}, f_{j_2}, \ldots, f_{j_{n-r}})$ where $\{j_1, j_2, \ldots, j_{n-r}\} = \{0,1,\ldots, n-1\} / \{i_1,i_2, \ldots, i_r\}$. Then $GH\T = 0_{r\ti (n-r)}$.

Both  $G$ and $H$ have full rank. 

When the  first $r$ rows chosen this gives
$\begin{pmatrix} e_0 \\ e_1 \\ \vdots \\ e_r \end{pmatrix} (f_0,
f_{n-1}, f_{n-2}, \ldots, f_{n-r}) = 0_{r\ti (n-r)}$ for the code system
expressing the generator and check matrices. 

\quad

\subsection{Vandermonde/Fourier matrices} 

 When  the rows are chosen from  Vandermonde/Fourier matrices and taken
 in arithmetic sequence with  arithmetic difference $k$ satisfying
 $\gcd(n,k)=1$ then MDS codes are obtained.  In particular when $k=1$,
 that is when  the rows are taken consecutively, MDS codes are
 obtained. This follows from results in   \cite{hurley} and these are
 explicitly recalled  in Theorems \ref{genthmfour}, \ref{genthmvan} below. 

The  $n\ti n$   Vandermonde matrix $V(x_1,x_2,\ldots, x_n)$ is defined by

$V=V(x_1,x_2,\ldots,x_n) = \begin{pmatrix}
1&1&\ldots &1 \\ x_1& x_2& \ldots& x_n \\ \vdots & \vdots &
\vdots & \vdots \\ x_1^{n-1} & x_2^{n-1} & \ldots &
x_n^{n-1} \end{pmatrix}$ 

 As is well known, the determinant of $V$ is $\prod_{i<j}(x_i-x_j)$.
Thus $\det(V) \neq 0$ if and only the $x_i$ are distinct.

A primitive $n^{th}$ root of unity $\om$ in a field $\F$ is an element
$\om$ satisfying $\om^n = 1_{\F}$ but $\om^i \neq 1_{\F}, 1\leq i <
n$. Often $1_\F$ is written simply as $1$ when the field is
clearly understood. 
   
The field $GF(q)$ (where $q$ is necessarily a power of a prime) contains
a primitive $(q-1)$ root of unity, see \cite{blahut,mceliece} or any
book on field theory, and such a
root  is referred to as a {\em primitive element in the field
$GF(q)$}. Thus also the field $GF(q)$ contains a primitive $n^{th}$ roots
of unity for any $n/(q-1)$.  

A Fourier $n\ti n$ matrix over $\F$ is a special type of Vandermonde matrix in which $x_i=\om^{i-1}$ and  $\om$ is a primitive $n^{th}$ root of unity in $\F$. Thus: 

$ F_n= \begin{pmatrix}1 & 1 & 1& \ldots & 1 \\ 1 & \om & \om^2 & \ldots &
	 \om^{n-1} \\ 
1 & \om^2 & \om^4 & \ldots & \om^{2(n-1)} \\ \vdots & \vdots & \vdots &
    \ldots & \vdots \\ 1 & \om^{n-1} & \om^{2(n-1)} & \ldots &
    \om^{(n-1)(n-1)} \end{pmatrix}$

is a Fourier matrix over $\F$ where $\om$ is a primitive $n^{th}$ root of unity in $\F$. 


Then  $$\begin{pmatrix}1 & 1 & 1& \ldots & 1 \\ 1 & \om & \om^2 & \ldots & 
	 \om^{n-1} \\ 
1 & \om^2 & \om^4 & \ldots & \om^{2(n-1)} \\ \vdots & \vdots & \vdots &
    \ldots & \vdots \\ 1 & \om^{n-1} & \om^{2(n-1)} & \ldots &
    \om^{(n-1)(n-1)} \end{pmatrix} \begin{pmatrix}1 & 1 & 1& \ldots & 1 \\ 1
				    & \om^{n-1} & \om^{2(n-1)} & \ldots
				    & \om^{(n-1)(n-1)}\\ 
1 & \om^{n-2} & \om^{2(n-2)} & \ldots &\om^{(n-1)(n-2)} \\ \vdots & \vdots & \vdots &
    \ldots & \vdots \\ 1 & \om & \om^{2} & \ldots &
    \om^{(n-1)} \end{pmatrix} = nI_n$$

Hence  $F_nF_n^*= nI_n$ where $F_n^*$ denotes the  second matrix on the left of the equation.
 Replacing $\om$ by $\om^{n-1}$ in $F_n$ is seen to give this $F_n^*$ 
 which  itself is a Fourier matrix. 
Refer to section \ref{type} for results on which fields  contain an
$n^{th}$ root of unity but in any case an $n^{th}$ root of unity can only 
exist in a field whose characteristic does not 
divide $n$.  

The following theorem on deriving MDS codes from Fourier matrices by unit-derived scheme is contained in \cite{hurley}:

\begin{theorem}\label{genthmfour}\cite{hurley} 

(i) Let $F_n$ be a Fourier $n\ti n$ matrix over a field $\F$. 
 Let $\mathcal{C}$ be the
 unit-derived code obtained by choosing in order $r$ rows of $V$ in
 arithmetic sequence with arithmetic 
difference $k$ and  $\gcd(n,k) = 1$. Then
 $\mathcal{C}$ is an MDS $(n,r,n-r+1)$ code. In particular this is true when
 $k=1$, that is when the $r$ rows are chosen in succession. 

(ii) Let $\mathcal{C}$ be as in part (i). Then there exist efficient encoding and decoding algorithms for $\mathcal{C}$.

\end{theorem}  

There is a similar, more general in some ways,   theorem for Vandermonde matrices:

\begin{theorem}\label{genthmvan}\cite{hurley} Let $V=V(x_1,x_2, \ldots, x_n)$ be a Vandermonde $n\ti n$ matrix over a field $\F$ with distinct and non-zero $x_i$. Let $\mathcal{C}$ be the unit-derived code obtained by choosing in order $r$ rows of $V$ in arithmetic sequence with difference $k$. If $(x_ix_j^{-1})$ is not a $k^{th}$ root of unity for $i\neq j$ then $\mathcal{C}$ is an $(n,r,n-r+1)$ mds code over $\F$. 

In particular the result holds for consecutive rows as then  $k=1$ and 
 $x_i\neq x_j$ for $i \neq j$.  
\end{theorem}

These are fundamental results. 

For  `rows in sequence' in the Fourier matrix cases, Theorem \ref{genthmfour},  and for some Vandermonde cases,  it is permitted that rows may {\em wrap around}
and then $e_k$ is taken to mean $e_{(k \mod n)}$. Thus for example Theorem \ref{genthmfour}
could be applied to a code generated by $\langle e_r, \ldots, e_{n-1},
e_0, e_1, \ldots, e_s\rangle$ where  $\langle e_0, e_1, \ldots,
e_{n-1}\rangle$ are the rows in order of a Fourier matrix. 

The general Vandermonde case is more difficult to deal with in practice but in any case using Fourier matrices is sufficient for coding purposes.  

Decoding methods for the codes produced  are given in the algorithms in
\cite{hurley} and in particular these are particularly nice for
the codes from Fourier matrices. 
The decoding methods are based on the decoding schemes  derived  in
\cite{hurleycomp} in connection with {\em compressed sensing} for solving underdetermined
systems using error-correcting codes. These decoding methods  themselves
are based on the error-correcting methods due to  Pellikaan \cite{pell}
which is a method of finding error-correcting pairs -- error-correcting pairs are shown to exist for the constructed codes and efficient decoding  algorithms are derived from this. These decoding algorithms are explicitly written down in detail
in \cite{hurley}. 
In addition the encoding itself is straightforward. 

  The complexity of encoding and decoding is  $\max \{ O(n\log
n), O(t^2)\}$ where $t= \floor{\frac{n-r}{2}}$;    
$t$ is the error-correcting capability of the code. The complexity is
given in Section \ref{complexity} and is derived in \cite{hurley}.

Let $F_n^*$ denote  the matrix with $F_nF_n^*= nI_{n\ti n}$ for the
Fourier matrix $F_n$. 
Denote the rows of $F_n$ in order  by $\{e_0,e_1, \ldots, e_{n-1}\}$ and denote 
the columns of $F_n^*$ in order by $\{f_0,f_1, \ldots, f_{n-1}\}$. {\em Then it is important to note that}  
$f_{i} = e_{n-i}\T, e_i = f_{n-i}\T$ with the convention that suffices
are taken modulo $n$. 
Also note $e_if_{i}= n$ and $e_if_j = 0, i\neq j$.
  

 Thus 

$$\begin{pmatrix}e_0 \\ e_1 \\ \vdots \\ e_{n-1}\end{pmatrix} (f_0,
f_1, f_2, \ldots, f_{n-1})= \begin{pmatrix}e_0 \\ e_1 \\ \vdots \\ e_{n-1}\end{pmatrix} (e_0\T,
e_{n-1}\T, e_{n-2}\T, \ldots, e_1\T) = nI_n$$




\subsection{Complexity}\label{complexity}
Efficient encoding and decoding algorithms exist for these codes by the
methods/algorithms  developed in \cite{hurley} which follow from those developed in \cite{hurleycomp}. In general the complexity is  
$\max\{O(n\log n), O(t^2)\}$ where $n$ is the length and $t$ is the
error-correcting capability, that is,  $t= \floor{\frac{d-1}{2}}$ where
$d$ is the distance. 
See the algorithms in \cite{hurley} for details;  there the decoding algorithms are derived  and are  written down precisely in suitable format. The decoding algorithms  reduce to finding a solution to a Hankel $t\ti (t+1)$ systems, which can be done in $O(t^2)$ time at worst,  and the other encoding and decoding algorithms are matrix multiplications which  can be reduced to multiplication  by a Fourier matrix which  takes  $O(n\log n)$ time.

\section{Maximum distance separable  codes}\label{general}  
\subsection{Given $n,r$}\label{general2}Suppose it is required to construct MDS  $(n,r)$ 
codes for given $n$ and $r$. 
First construct a $n\ti n$ Fourier matrix
over a finite field. A Fourier $n\ti n$ matrix is constructible over a
finite field of
characteristic $p$ where $p\,\not\vert \, n$, see  section \ref{type}.  
 Take $r$
rows in sequence with arithmetic difference $k$ satisfying
$\gcd(n,k)=1$ from this Fourier matrix. Then by Theorem
\ref{genthmfour}, see \cite{hurley} for details,  the code generated by
these rows is an $(n,r)$ MDS code. There are many different
ways for constructing  the $(n,r)$ code from the Fourier $n\ti
n$ matrix -- one could start at any row with $k=1$ and could also start
at any row for any $k$ satisfying $(n,k)=1$. A check matrix may be read
off immediately from section \ref{construct} and a direct decoding
algorithm of complexity $\max\{O(n \log n), O(t^2)\}$ is given in 
\cite{hurley}, where $t$  is the error-correcting capability.  

\subsection{MDS to required rate and error-correcting capability}\label{required}
Suppose it is required to construct an MDS  code of  rate $R$ 
 and to required error-correcting capability. The required code is of
 the form $(n,r)$ with $(n-r+1) \geq (2t+1)$ where $t$ is the required
 error-correcting capability. Now  $(n-r+1) \geq (2t+1)$ requires 
 $n(1-R) \geq 2t$. {\em Thus require $n\geq \frac{2t}{1-R}$}. With these
 requirements construct the Fourier $n\ti n$ and from this take $r\geq  nR$
 rows in arithmetic sequence with arithmetic difference $k$ satisfying
 $\gcd(n,k)=1$. The code constructed has the required parameters. The
 finite fields over which this Fourier matrix can be constructed 
 is  deduced from section \ref{type}.

\paragraph{Samples} It is required to construct a rate $R=\frac{7}{8}$
 code which can correct $25$ errors. Thus, from general form $n\geq
 \frac{2t}{1-R}$, require $n\geq
 \frac{50}{\frac{1}{8}}$ and so $n\geq 400$. 

Consider $n=400$. Construct a Fourier $400 \ti 400$ matrix $F_{400}$ over a
suitable finite field. Then $r= 350$ for rate $\frac{7}{8}$. 
Now take any $350$ rows in sequence from $F_{400}$ with arithmetic
difference $k$ satisfying $\gcd(400,k)= 1$. Now $k=1$ starting at first row  works in any case
but there are many more which are suitable. The code generated by these rows
is an $(400,350,51)$ code, Theorem \ref{genthmfour},  which can correct $25$ errors. 

Over which fields can the Fourier $400\ti 400$ matrix exist? The
characteristic of the field must not divide $400$ but finite fields of any other characteristic exist over which the
Fourier $400 \ti 400$ matrix is constructible. For example: the order of $3
\mod 400$ is $20$ so $GF(3^{20})$ is suitable; the order of $7 \mod 400$
is $4$ so $GF(7^4)$ is suitable. Exercise: Which other fields are suitable?  

However $401$ is prime and  the order of $401 \mod 400$ is $1$ and thus
the prime field $GF(401)$ is suitable. It is also easy to find a
primitive $400$ root of unity in $GF(401)$; indeed the order of $\om = (3 \mod
401)$ is $400$ in $GF(401)$ and this element may be used to generate the $400\ti
400$ Fourier matrix over $GF(401)$. The arithmetic is modular arithmetic in
$\Z_{401} = GF(401)$. 
\footnote{The Computer Algebra system GAP \cite{gap}  
has the command \ord$(r,m)$ which is useful. This system also has the coding package GUAVA with which experiments can be made. }

A field of characteristic $2$ close to the requirements may be
prescribed. Then let  $n=399$ and note that the order of $2 \mod 399$ is
$18$. Thus use the field $GF(2^{18})$ over which the Fourier $399 \ti
399$ matrix may be constructed. Take $348$ rows of this Fourier matrix in sequence
with arithmetic difference $k$ satisfying $(399, k)=1$ to form a $(399,348, 51)$ code which can correct
$25$ errors. Rate is $0.8746 ..$ which is close to required rate
$\frac{7}{8}$.  

Exercise: How many (different) $(n,r)$ MDS codes may be formed from this Fourier $n\ti n$
matrix over the finite field? Note the sequence may `wrap over' and then
the numbering is mod $n$.

\subsection{Infinite series with given rate}\label{infinite} Construct an infinite
series of codes with given rate $R$ such that the limit of the distance
by the length approaches $(1-R)$. 

Let $R= \frac{r}{n}$ be given.  Construct the Fourier $n\ti n$ matrix and
from this derive the $(n,r)$ MDS code as in section \ref{construct}. Let 
$n_i=i*n,
r_i=i*r$ for an increasing set of positive integers $\{i\}$. Construct the Fourier $n_i\ti n_i$ matrix and from this derive
an $(n_i,r_i)$ MDS code. The rate of the code is $\frac{r_i}{n_i} =
\frac{r}{n} = R$. The distance of the code is $d_i= (n_i-r_i+1)$. The
ratio of the length by the distance is $\frac{n_i-r_i+1}{n_i} = 1-R +
\frac{1}{n_i}$. Now as $i\rightarrow \infty$ it is seen that the ratio
of the distance by the length approaches $(1-R)$. 

Note that $0 < R < 1$ if and only if $0< (1-R) < 1$ so could start off
with a requirement that the limit approaches a certain fraction. 

There are many choices by this method giving  different series. At
each stage there are many different $n_i\ti n_i$ Fourier matrices to
choose from and within each of these are many choices of $r_i$ rows for
obtaining $(n_i,r_i)$ MDS codes.
   
By methods/algorithms of \cite{hurley} the codes have efficient 
encoding and decoding algorithms of complexity $\max\{O(n \log n),
t^2\}$ where $t$ is the error-correcting capability. 

  
\subsection{Series in characteristic  $p$ with given rate}\label{charp} Suppose codes over fields of characteristic $2$ are required. Now a
Fourier matrix of even size in characteristic $2$ cannot
exist. It is necessary to consider rates of the form 
$\frac{r}{n}$ where $n$ is odd in order for the general method of section \ref{infinite} to work in   characteristic $2$. The method of section \ref{infinite} is then 
applied by taking the increasing sequence $\{i\}$ to consist of odd elements only. Then construct the Fourier $(n*i)
\ti (n*i)$  matrix for odd $i$ (and odd $n$) in a field of 
characteristic $2$ -- see section \ref{type} on method  to form such a
Fourier $n*i\ti n*i$ matrix in a finite field of characteristic $2$. From this Fourier matrix construct an MDS
$(n*i,r*i)$ code with rate $\frac{r}{n}$ by method of Theorem \ref{genthmfour}; there are choices for this code as noted.  

As a sample consider the  rate 
$\frac{7}{9}$. Then Fourier $(9*i) \ti (9*i)$ matrices are constructible
over fields of characteristic $2$ for odd $i$.  From this 
$(9*i,7*i, 2*i+1)$ codes are constructed. 

Thus $(9,7,3)$ code over $GF(2^6)$, $(27,21,7)$ code over
$GF(2^{18})$, $(45, 35,11)$ code over $GF(2^{12})$, \\ $(63, 49,15)$ over
$GF(2^6)$, and so on, are constructed. The fields of characteristic $2$ used depend on
the order of $2 $ modulo the required length. 
The ratio of the distance by the length approaches $(1-R) = \frac{2}{9}$.

\quad

Similarly infinite series of codes over fields of characteristic $p$ are
constructed with given rate $\frac{r}{n}$ where 
$p \not\vert \, n$.
\paragraph{Sample} For example consider rate $R= \frac{7}{10}$ for   characteristic $3$. Then 
the method constructs MDS \\ $\{(10,7,4), (20,14,7),
(40,28,13), (50,35,16), (70, 49, 22), (80,56,25), \ldots \}$ codes in
fields of characteristic $3$. 

Now \ord$(3,10)= 4$, \ord$(3,20)=4$,
\ord$(3,40)=4$, \ord$(3,50)=20$, \ord$(3,70)= 12$, \ord$(3,80)=4 $, \ldots 
so these codes can be constructed  respectively over \\ $\{GF(3^4), GF(3^4),
GF(3^4), GF(3^{20}), GF(3^{12}), GF(3^4), \ldots \}$. It is seen that
$(80,56,25)$ is over a relatively small field $GF(81)$ and  can correct
$12$ errors. 
 The limit of
the distance over the length is $(1-R)=\frac{3}{10}$.    

\subsubsection{Note} In characteristic $p$ the rates $\frac{r}{n}$ attainable require $p\not\vert \, n$ so  that  the Fourier matrix $n\ti n$ can be constructed in characteristic $p$. This is not a great restriction. For any given fraction $R$ and any given $\ep > 0$ there exists a fraction with numerator not divisible by $p$ between $R$ and $R+\ep$. The details are omitted. For example suppose in characteristic $2$ the rate required is $\frac{3}{4}$   and $\ep >0$ is given. Say $\frac{1}{32} < \ep$ and then need a fraction of the required type between $\frac{3}{4}$ and $\frac{3}{4} + \frac{1}{32} = \frac{25}{32}$. Now $\frac{24}{31}$ will do and we can proceed with this fraction to construct the codes over characteristic $2$; the Fourier $31\ti 31$ matrix exists over $GF(2^5)$.   
\subsection{Infinite series in prime fields with given rate}\label{primes} Arithmetic in prime fields is particularly nice. Here we develop a method for constructing series of MDS codes over prime fields.

Suppose a rate $R$ is required, $0<R<1$. Let $p$ be a prime and consider
the field $GF(p)= \Z_p$. This has an element of order $(p-1)$ and thus
construct the Fourier $(p-1) \ti (p-1)$ matrix $F_{p-1}$ over $GF(p) =
\Z_{p}$. For this it is required to find a primitive $(p-1)$ root of
unity in $GF(p)=\Z_p$. \footnote{It seems there is no known algorithm for
      finding a generator of $(\Z_p/\{0\})$ that is substantially better
      than a brute force method - see Keith Conrad's notes \cite{conrad}. Note however there are precisely $\phi(p-1)$ generators.}
Let $r=\floor{(p-1)*R}$. Now $p$ must be large enough so that
$r\geq 1$. Form the $(p-1, r)$ MDS code over $F_{p-1}$. This has rate close
to $R$. 

Let $\{p_1, p_2, \ldots, p_i, \ldots \}$ be an infinite increasing set
  of primes such that $(p_1-1)*R \geq 1$ in which case $(p_i-1)*R \geq  1$ for each $i$. Form the Fourier $(p_i-1)\ti
  (p_i-1)$ matrix over $GF(p_i)$. Let $r_i= \floor{(p_i-1)*R}$. Form the
  $(p_i-1, r_i)$ MDS code over $GF(p_i)$. The ratio of the distance to
  the length is $\frac{p_i-1- r_i+1}{p_i-1} = 
  1- \frac{r_i}{p_i-1} + \frac{1}{p_i-1}$. Now as $i\rightarrow \infty$
  this ratio approaches $(1-R)$. 

  \paragraph{Sample} Let $\{p_1, p_2, \ldots \}$ be the  primes of the
  form $(4n+1)$ 
  and let $R=\frac{3}{4}$. Now $p_1=5,p_2= 13, p_3=17, \ldots$. Let $r_i= (p_i-1)*R = 4*j*\frac{3}{4} = j*3$ for some $j$. Form the $(p_i-1, r_i)$ code from Fourier $(p_i-1) \ti (p_i-1)$ matrix over $GF(p_i)$.

  Get codes $\{(4,3,2), (12, 9, 4), (16, 12, 5), (28,21, 8), (36,
  27,10), \ldots \}$ over, respectively, the following fields \\ $\{GF(5), GF(13), GF(17), GF(29), GF(37), \ldots \}$. 


  \subsection{Sample of the workings}\label{sample} Here is an example of MDS codes in
  $GF(13)= \Z_{13}$. A primitive element in $GF(13)$ is $\om = (2 \mod
  13)$. The Fourier $12 \ti 12$ matrix with this $\om$ as the element of
  order $12$ in $GF(13) = \Z_{13}$ is:

  $F_{12} = \left( \begin{smallmatrix}1& 1& 1& 1& 1& 1& 1& 1& 1& 1& 1& 1 \\ 
   1& 2& 4& 8& 3 &6 &12 &11 &9 &5 &10 &7 \\ 
   1& 4& 3 &12& 9 &10 &1 &4 &3 &12& 9 &10  \\ 
  1&8 &12& 5& 1 &8 &12 &5 &1 &8 &12 &5 \\ 
   1 &3& 9& 1& 3& 9& 1& 3 &9 &1 &3 &9 \\ 
   1& 6 &10 &8 &9 &2 &12& 7& 3& 5 &4 &11 \\ 
   1 &12& 1& 12& 1& 12& 1& 12& 1 &12& 1& 12 \\ 
   1 &11 &4 &5 &3 &7 &12 &2 &9 &8 &10& 6  \\ 
  1& 9& 3& 1 &9 &3& 1& 9& 3& 1& 9&3 \\ 
   1 &5 &12& 8& 1& 5& 12& 8 &1& 5& 12& 8 \\ 
   1 &10& 9 &12& 3& 4& 1 &10& 9& 12& 3& 4 \\ 
   1& 7& 10& 5& 9& 11& 12& 6& 3& 8& 4& 2 
  \end{smallmatrix}\right)$

  Let the rows of $F_{12}$ in order be denoted by $\{e_0, e_1, \ldots, e_{11}\}$.
  
  Various MDS codes over $GF(13)$ may be constructed from $F_{12}$.

  Two of the $(12,6,7)$ codes are as  follows:  
    
$$K=\left[\begin{smallmatrix} 1& 1& 1& 1& 1& 1& 1& 1& 1& 1& 1& 1  
\\ 1& 2&4& 8& 3& 6& 12& 11& 9& 5& 10& 7 
\\ 1& 4& 3& 12& 9& 10& 1& 4& 3& 12& 9& 10 
\\ 1& 8& 12& 5& 1& 8& 12& 5& 1& 8& 12& 5  
\\ 1& 3& 9& 1& 3& 9& 1& 3& 9& 1& 3&9  
\\ 1 & 6& 10& 8& 9& 2& 12& 7& 3 &5& 4&11 \end{smallmatrix}\right], \, L=\left[\begin{smallmatrix} 1& 2& 4& 8& 3 &6 &12 &11 &9 &5 &10 &7 \\
        1 &12& 1& 12& 1& 12& 1& 12& 1 &12& 1& 12 \\ 1& 7& 10& 5& 9& 11& 12& 6& 3& 8& 4& 2 \\  1 &3& 9& 1& 3& 9& 1& 3 &9 &1 &3 &9 \\  1 &5 &12& 8& 1& 5& 12& 8 &1& 5& 12& 8 \\  1& 4& 3 &12& 9 &10 &1 &4 &3 &12& 9 &10  \end{smallmatrix}\right]$$
    
The first matrix  takes the first $6$ rows of $F_{12}$; the second matrix  takes rows  $\{e_1, e_6, e_{11}, e_4, e_9,e_2\}$ which are  $6$ rows in sequence with arithmetic difference $5$, $\gcd(12,5) = 1$,   starting with the second row.  
These are  generator matrices for  $(12,6,7)$ codes over $GF(13) =\Z_{13}$ and each  can correct $3$ errors.  

    \subsubsection{Correcting errors sample}\label{err}  Efficient decoding algorithms for
    the codes  are established in \cite{hurley}.   Here is an example to
    show how the algorithms work in practice.  The matrix $K$ as above, formed from  the first
    $6$ rows of  Fourier matrix $F_{12}$, is the generator matrix of a $(12,6,7)$ code. 
    Apply Algorithm 6.1 from \cite{hurley} to correct up to $3$
    errors of the code as follows. Note the work is done in $\Z_{13} = GF(13)$ using  modular arithmetic.  
    \begin{enumerate} \item The word $\un{w}=( 8, 9, 2, 6, 3, 3, 10, 8, 4, 1, 5, 7 )$ is received.
      \item Apply check matrix to $\un{w}$ and get $\un{e}= (2, 9, 12, 10, 11, 11 )$. Thus there are errors and $\un{w}$ is not a codeword. (The check matrix   $(e_1\T,e_2\T, e_3\T,e_4\T,e_5\T,e_6\T)$ is immediate, see section \ref{construct}.)
      \item Find a non-zero element of the kernel of $\begin{pmatrix}
						  2&9&12&10 \\ 9&12&10&
						  11 \\ 12& 10& 11&11
						  \end{pmatrix}$. This
	    is a $3\ti 4$ Hankel matrix, formed from $\un{e}$; the first row
	    consists of elements $(1-4)$ of $\un{e}$, the second row
	    consists of elements  $(2-5)$ of $\un{e}$, and the third row consists of
	    elements 
	    $(3-6)$ of $\un{e}$.  A non-zero element of the 
	    kernel is  $\un{x}= (7,1,7,1)\T$.
      \item Now  $\un{a} = (e_1,e_2,e_3,e_4)*\un{x} = (3, 12, 7, 0, 1, 0, 1, 2, 4, 0, 10, 12 )$. Thus errors occur at $4^{th},6^{th}, 10^{th}$ positions (which are the positions of the zeros of $\un{a}$).
      \item Solve  (from $4^{th}, 6^{th}, 10^{th}$ columns of $(2-7)$ rows of $F_{12}$): 
        
        $\begin{pmatrix} 8& 6 & 5 \\ 12 & 10 &12 \\ 5&8&8 \\ 1&9&1 \\ 8&2&5\\ 12& 12 & 12\end{pmatrix}\begin{pmatrix}x_1 \\ x_ 2 \\ x_3 \end{pmatrix} = \begin{pmatrix} 2 \\9\\ 12\\ 10 \\ 11 \\ 11 \end{pmatrix}$

        In fact only the first three equations need be solved; answer is $(10,1,4)\T$. Thus error vector is $\un{k}= (0,0,0,10,0,1,0,0,0,4,0,0)$.
      \item Correct codeword is $\un{c}= \un{w}- \un{k} = (8, 9, 2, 9, 3, 2, 10, 8, 4, 10, 5, 7)  $.
      \item If required, the original data word can be obtained directly
    by multiplying by the right inverse of  the generator matrix; the
    right inverse is read off as  $K= (e_0,
    e_{11},e_{10},e_9,e_8,e_7)\T*12 $.  Then $\un{c}*K = (1,2,3,4,5,6)$
    which is the original data word to be safely  transmitted. \end{enumerate}

    The equations to be solved are Hankel matrices of size the order of $t\ti t$ where $t$ is the error-correcting capability.


\subsection{Length $(2^q-1)$ MDS codes  in $GF(2^q)$}\label{length} 
$2^3-1 =7, 2^4-1 = 15, 2^5-1=31, 2^6-1 = 63, \ldots$. 

In general consider the characteristic $2$ field $GF(2^q)$. In this
field acquire an element of order $n=(2^q-1)$ and construct the Fourier
$n\ti n$ matrix over $GF(2^q)$. From this,  MDS $(n, r)$ codes are
constructed for $1\leq r \leq n$. It is better  to take odd $r$ from
consideration of the error-correcting capability. 
\begin{enumerate}
\item $2^3-1=7$. From the Fourier $7\ti 7$ matrix over $GF(2^3)$
      construct the MDS $\{(7,5,3),(7,3,5),(7,1,7)\}$ codes which can
      correct respectively $\{1,2,3\}$ errors.
    \item $2^4-1= 15$.
From the Fourier $15\ti 15$ matrix over $GF(2^4)$
      construct \\ $\{(15,13,3), (15,11,5),(15,9,7), (15,7,9),(15,5,11),
      (15,3,13), (15,1, 15)\}$ MDS codes over $GF(2^4)$ which can correct
      respectively $\{1,2,3,4,5,6,7\}$ errors. 
    \item $2^5-1 = 31$.
      From the Fourier $31\ti 31$ matrix over $GF(2^5)$
      construct the \\ MDS  $\{(31,29, 3), (31,27,5), (31,25,7), \ldots,
      (31,3,29), (31,1,31)\}$ codes which can respectively correct
      $\{1,2,3,\ldots, 14,15\}$ errors.

\item $2^8 -1 = 255$. Thus MDS codes $(255,r)$ are constructed over
      $GF(2^8)$ for all $r$. 
  These could be compared to Reed-Solomon codes used in practice and perform better. 

Even further consider $2^9-1= 511$. Then MDS codes $(511,r)$ are constructed over $GF(2^9)$. For example  $(511,495,17), (511,487,25)$ codes are
      constructed over $GF(2^9)$; the decoding algorithm involves
      finding a solution to $9\ti 8$, $13 \ti 12$ (respectively) Hankel
      systems of equations,  and matrix Fourier multiplication.   

The codes over prime fields  in section \ref{prime} of length $256$ over
      $GF(257)= \Z_{257}$ and of length $508$ over $GF(509)=\Z_{509}$
      perform better.   
\item \ldots \ldots
\item General $2^q-1=n$. From the Fourier $n\ti n$  matrix over
      $GF(2^q)$ construct the MDS $(n,n-2,3), (n,n-4,5), (n,n-6, 7),
      \ldots, (n,n-2m, 2m+1), \ldots, (n,3,n-2), (n,1,n)$ codes which
      can correct respectively $\{1,2,3,\ldots, m, \ldots, \frac{n-3}{2}, 
\frac{n-1}{2}\}$ errors. 
\end{enumerate} 


It is clear that similar series of relatively large length MDS codes may
be constructed over finite fields of characteristics other than $2$. 

\subsection{Length $(p-1)$ codes in prime field $GF(p)=\Z_p$}\label{prime}
Construct large length MDS codes over prime fields. This is a particular
general case of section \ref{length} but is singled out as the
arithmetic involved, modular arithmetic,  is smooth and very  efficient and the examples are nice and practical.  
 For any prime $p$ the Fourier $(p-1)\ti (p-1)$ matrix exists
over $GF(p) =\Z_{p}$. A primitive $(p-1)$ root of
unity is required in $GF(p)$ \footnote{It seems there is no known algorithm in which
to find  a generator of $(\Z_p/\{0\})$ that is substantially better
      than a brute force method - see Keith Conrad's notes
      \cite{conrad}. Note however there are precisely $\phi(p-1)$
      primitive $(p-1)$ roots of unity in $GF(p)=\Z_p$.}. 
The arithmetic is modular
arithmetic in $\Z_p$ which is nice.  The general method then allows the
construction of MDS $(p-1,r)$ codes over $GF(p)$ for any $1 \leq r \leq
(p-1)$. It is better  to use even  $r$,  so that the distance is then odd -- for $p>2$.  

Here are samples: 

\begin{enumerate} 
\item $p=11$. Then MDS codes of the form $\{(10,8,3), (10,6, 5), (10,4,7),
      (10,2,9)\}$ are constructed over $GF(11)=\Z_{11}$. They can
      respectively correct $\{1,2,3,4\}$ errors. A primitive $10^{th}$ root
      of unity is $(2 \mod 11)$; also $(7 \mod 11)$ is a primitive
      $10^{th}$ root of unity. The method allows the construction of (at
      least) $\phi(11) = 10$ MDS $(12,r)$ codes for each $r$. 
      
\item $p=13$. Then MDS codes of the forms $\{(12,10,3),(12,8,5),(12,6,7),
      (12,4,9), (12,2,11)\}$ are constructed  over $GF(13) =\Z_{13}$ which
      can correct respectively $ \{1,2,3,4,5\}$ errors. A primitive
      $12^{th}$ root of unity is $(2 \mod 13)$ or $(7 \mod 13)$. 
  \item $p=17$. Then MDS codes of the forms

	$\{(16,14,3),(16,12,5),(16,10,7),(16,8,9),
	(16,6,11),(16,4,13),(16,2,15)\}$ which can correct respectively
	$\{1,2,3,4,5,6,7\}$ errors are constructed over $GF(17)=
	\Z_{17}$. A primitive $16^{th}$ root of unity in $GF(17)$ is $ (3 \mod 17)$
	or $(5 \mod 17)$ and there are $\phi(16) = 8$ such generators. 
\item \ldots \ldots
\item\label{compare} {Relatively large sample with modular arithmetic: for comparison.}
Consider $GF(257) =\Z_{257}$ and $257$ is prime. 
Construct the Fourier matrix $F_{256}$ with a primitive $256^{th}$ root of
      unity $\om$ in $GF(257)$. Since the order of $3 \mod 257$ is $256$
      then a choice for $\om$ is $ (3 \mod 257)$. Denote the rows of
      $F_{256}$ in order by $\{e_0, e_1, \ldots, e_{255}\}$.

Suppose a dimension $r$ is required. Choose $\C = \langle e_0,
      e_1, \ldots, e_{r-1}\rangle $ to get an MDS $(256,r)$
      code. The arithmetic is modular arithmetic, $\mod 257$, and work
      is done   with powers of $(3 \mod 257)$. In addition  $(5 \mod 257)$ or $(7 \mod 257)$ could  be
      used to generate  the Fourier $256\ti 256$ matrix over
      $GF(257)=\Z_{257}$; indeed there exist $\phi(256)= 128$ generators
      that could be used to generate the Fourier matrix.

Note that $(256,240,17)$ and $(256, 224,23)$ codes over $GF(257)$ are
      constructed as well as other rate codes. These particular ones
      could be compared to the Reed-Solomon $(255,239,17)$ and
      $(255,223,23)$ codes which are in practical use; the ones from $GF(257)$ perform better and faster. There is a
      much bigger choice for rate and error-correcting capability.  

Bigger primes could also be used. Taking $p=509$ gives $(508,r)$ MDS
      codes for any $1<r<508$. Thus for example $(508, 486,23)$ MDS
      codes over $GF(509)=\Z_{509}$ are constructed. 

      The method allows the construction of $\phi(256)=128$ such MDS
      $(256, r)$ codes with different generators for the Fourier
      matrix. For larger primes the  number that could be used for the 
      construction of the Fourier matrix is substantial and
      cryptographic methods could   be
      devised from such considerations. For example for the prime $p=2^{31}-1$ the Fourier
      $(p-1)\ti (p-1) $ matrix exists over $GF(p)$ and 
      $\phi(p-1)=534600000$ elements could be used to generate the
      Fourier matrix. 
  \item The $p$ can be very large and the arithmetic is still
	doable. For example $p=10009$ allows the construction of
	$(10008,r)$ MDS codes over $GF(10009) = \Z_{10009}$. If $100$
	errors are required to be corrected the scheme supplies
	$(10008, 9808, 201)$ MDS codes over $GF(10009) = \Z_{10009}$
	which have large rate $\approx .98$ and can correct $100$ errors. The arithmetic is modular
	arithmetic. The order of $\om = (11 \mod 10009)$ is $10008$ so
	this $\om$ could be used to generate the Fourier $10008\ti 10008$
	matrix over $GF(10009)=\Z_{10009}$; indeed there are
	$\phi(10008)= 3312$ different elements in $GF(10009)= \Z_{10009}$ that could be used to
	generate the Fourier $10008\ti 10008$ matrix.   
\item General $p$. Then MDS codes of the form $(p-1,p-3,3),(p-1,p-5,5),
      (p-1,p-7,7), \ldots, (p-1, p- (2i+1),2i+1), \ldots, (p-1,2,p-2)$ are
      constructed which can respectively correct $\{1,2,3,\ldots, i,
      \ldots \frac{p-3}{2}\}$ errors are constructed. A primitive modular
      element (of order $(p-1)$) is obtained in $GF(p)=\Z_p$ with which to construct the Fourier matrix; as already noted it seems a brute force method for obtaining such seems to be as good as any. 
\end{enumerate}
\subsection{The  fields}\label{type} Suppose $n$ is given and it is required 
to find finite fields over which a Fourier $n\ti n$ matrix exists. The
following argument is essentially taken from \cite{hurley}. It is included for clarity and completeness and is necessary for deciding on the relevant fields to be used in cases.

Note first of
all that the field must have characteristic which does not divide $n$ in
order for the Fourier $n\ti n$ matrix to exist over the field.
\begin{proposition} There exists a finite field of characteristic $p$
 containing an $n^{th}$ root of unity for given $n$ if and only if
 $p\not\vert \, n$. 
\end{proposition}
\begin{proof} Let
$p$ be a prime which does not divide $n$.  Hence  $p^{\phi(n)} \equiv 1
\mod n$ by Euler's theorem where $\phi$ denotes the Euler $\phi$ function. More specifically let $\be$ be the least
positive integer such that $p^{\be} \equiv 1
\mod n$. Consider $GF(p^\be)$. Let $\de$ be a primitive element in
$GF(p^\be)$. Then $\de$ has order $(p^\be-1)$ in $GF(p^\be)$ and $(p^\be -1) = s n$ for
some $s$. Thus $\om= \de^s$ has order $n$ in $GF(p^\be)$.

On the other hand if $p/n$ then $n=0$ in a field of characteristic $p$ and so no $n^{th}$ root of unity can exist in the field.
\end{proof}

The proof is  constructive. Let $n$ be given and $p\not\vert \, n$. Let $\be$ be the least power such that $p^\be \equiv 1 \mod n$; it is known that $p^{\phi(n)} \equiv 1 \mod n$ and thus $\be$ is a divisor of $\phi(n)$. Then the Fourier $n\ti n$ matrix over $GF(p^\be)$ exists.  
\paragraph{Sample}\label{fields}
Suppose $n=52$. The prime divisors of $n$ are $2,13$ so take any other prime $p$ and then there is a field of characteristic $p$ which contains a $52^{nd}$ root of unity. For example take $p=3$. Know $3^{\phi(52)} \equiv 1 \mod 52$ and $\phi(52) = 24$ but indeed $3^6 \equiv 1 \mod 52$. Thus the field $GF(3^6)$ contains a primitive $52^{nd}$ root of unity and the Fourier $52 \ti 52$ matrix exists in $GF(3^6)$. Also $5^4\equiv 1 \mod 52$, and so $GF(5^4)$ can be used.  Now $5^4 = 625 < 729 = 3^6$ so $GF(5^4)$ is a smaller field with  which to work.

Even better though is $GF(53)= \Z_{53}$ which is a prime field. This has an element of order $52$ from which the Fourier $52\ti 52$ matrix can be formed. Now $\om = (2 \mod 53)$ is an element of order $52$ in $GF(53)$. Work and codes with the resulting Fourier $52\ti 52$ matrix can then be done in modular arithmetic, within $\Z_{53}$, using powers of $(2 \mod 53)$.

 \subsubsection{Developments on different {\em types} of MDS codes that can be constructed}\label{further} 
This section is for information on developments and is not required subsequently. 
 
Particular {\em types} of MDS codes may be required. These are not dealt with here but the following is noted. 

\begin{itemize} 
\item  A quantum MDS code is one  of the form $[[n,r,d]]$ where $2d=n-r+2$, see \cite{ash} for details.  In \cite{quantum} the methods are applied to construct and develop MDS quantum
      codes of different types and to required specifications. This is done by requiring the constructed codes to be {\em dual-containing MDS codes} from which {\em quantum MDS
      error-correcting codes}  are constructed from  the CSS construction  developed in \cite{calderbank,good}.  

        This is further developed for the construction and development of   {\em Entanglement assisted quantum
      error-correcting codes}, EAQECC, of different types and to required 
      specifications  in \cite{eaqecc}.
\item 
  In \cite{LCD} Linear complementary dual (LCD), MDS codes
		      are constructed based on the general constructions. 
		      An LCD code $\C$ is a code such
		      that $\C\cap \C^\perp = 0$. These have found use
		      in security, in data storage and communications'
		      systems. In \cite{LCD} the rows are chosen according to
		      a particular formulation so as to derive LCD codes which are also MDS codes. 
                     
    \item In \cite{hurleycomp} error-correcting codes, similar to ones here,
      are used for solving
      {\em underdetermined systems of equations} for use in {\em compressed
      sensing}.
\item By using rows of the Fourier matrix as matrices for polynomials,  
  MDS convolutional codes, achieving the {\em generalized Singleton bound} see \cite{ros}, 
  are constructed and
  analysed in  \cite{hurconv2}.
  \item The codes developed here seem particularly suitable for use in
	McEliece type encryption/decryption, \cite{mceliece2}; this has
	yet to be investigated.  
\end{itemize}

\end{document}